\documentclass[letterpaper, 10 pt, conference]{ieeeconf}  

\IEEEoverridecommandlockouts                              
\usepackage{amsmath} 
\usepackage{amssymb}  

\usepackage{amsthm}
\usepackage{graphicx}
\usepackage{mathtools}
\usepackage{xcolor}
\definecolor{ao(english)}{rgb}{0.0, 0.5, 0.0}
\usepackage[colorlinks, citecolor = {ao(english)}, linkcolor = {ao(english)}]{hyperref} 
\usepackage[capitalize,nameinlink]{cleveref}[0.19]
\usepackage{float}
\usepackage{subfig}
\usepackage{pgfplots}
\pgfplotsset{compat=newest}

\usepackage{xfrac}

\usepackage{tikz}
\usepackage{tkz-euclide}

\usepackage{enumerate}

\usepackage{mathrsfs}

\usepackage{siunitx}
\usepackage[sort,nocompress]{cite}

\usetikzlibrary{arrows,shapes,trees,calc,positioning,patterns,decorations.pathmorphing,decorations.markings,backgrounds}
\usetikzlibrary{matrix}

\usepgfplotslibrary{groupplots}

\usepackage[linesnumbered, ruled]{algorithm2e}
\SetKwRepeat{Do}{do}{while}%

\newtheorem{theorem}{Theorem}

\newtheorem{lemma}[theorem]{Lemma}

\newtheorem{assumption}[theorem]{Assumption}

\newcommand{\p}{\partial}

\newcommand{\real}{\mathbb{R}}

\newcommand{\normal}[2]{\mathcal{N}\left({#1},{#2}\right)}

\newcommand{\ev}{\mathbb{E}}
\newcommand{\trace}[1]{\textup{tr} \ #1}

\newcommand{\blkdiag}{\textup{\texttt{blkdiag}}}

\newcommand{\sign}{\textup{sgn}}

\newcommand{\evcanon}{\ev_{\qinit,\pinit}}

\newcommand{\gradlog}{h}

\newcommand{\hpos}{q}
\newcommand{\hmom}{\rho}

\newcommand{\utol}{\delta_u}

\newcommand{\loopi}{k}
\newcommand{\utmp}{u^\dagger}

\newcommand{\usgn}{u_\textup{sgn}}

\newcommand{\pstcost}{J}
\newcommand{\epstcost}{\bar{J}}

\newcommand{\st}{\textup{s.t.}}
\newcommand{\defeq}{:=}

\newcommand{\mysec}{\S}

\newcommand{\disccol}{black} 

\newcommand{\discuss}[1]{\textcolor{\disccol}{#1}} 

\newcommand{\sys}{\mathcal{S}}
\newcommand{\hamsys}{\mathcal{H}}

\newcommand{\nq}{d}
\newcommand{\Inertia}{M}

\newcommand{\qmap}{\vec{q}}

\newcommand{\qinit}{\bar{\hpos}}
\newcommand{\pinit}{\bar{\hmom}}

\newcommand{\qnom}{q^*}

\newcommand{\stepsize}{\epsilon}

\newcommand{\simlength}{\ell}
\newcommand{\simsteps}{L}

\newcommand{\hmcparam}{{\gamma_\textup{hmc}}}

\newcommand{\myalg}{Alg.}
\newcommand{\mbf}[1]{\mathcal{I}_{#1}}

\newcommand{\mriparam}{\tau}

\newcommand{\udp}{u^{\textup{dp}}}

\newcommand{\maxpower}{\bar{U}}

\newcommand{\inind}{0}

\newcommand{\target}{\pi}

\newcommand{\posterior}{p}

\newcommand{\pstmean}{\mu}

\newcommand{\evpost}{\ev_{\param|y}}

\newcommand{\evdata}{\ev_{y|\paramtr}}

\newcommand{\uopt}{u^*}
\newcommand{\unom}{\tilde{u}}

\newcommand{\yopt}{y^*}
\newcommand{\ynom}{\tilde{y}}

\newcommand{\ydet}{\tilde{y}}
\newcommand{\xdet}{\tilde{x}}

\newcommand{\numy}{N}
\newcommand{\numq}{M}

\newcommand{\canonical}{p}

\newcommand{\numdp}{T}
\newcommand{\feasinput}{\mathcal{U}}

\newcommand{\nparam}{n_{\param}}

\newcommand{\weight}{W}

\newcommand{\umin}{\underline{u}}
\newcommand{\umax}{\bar{u}}

\newcommand{\param}{\theta}
\newcommand{\paramtr}{{\theta^*}}

\renewcommand{\baselinestretch}{0.99}

\crefformat{equation}{\textup{#2(#1)#3}}
\crefrangeformat{equation}{\textup{#3(#1)#4--#5(#2)#6}}
\crefmultiformat{equation}{\textup{#2(#1)#3}}{ and \textup{#2(#1)#3}}
{, \textup{#2(#1)#3}}{, and \textup{#2(#1)#3}}
\crefrangemultiformat{equation}{\textup{#3(#1)#4--#5(#2)#6}}%
{ and \textup{#3(#1)#4--#5(#2)#6}}{, \textup{#3(#1)#4--#5(#2)#6}}{, and \textup{#3(#1)#4--#5(#2)#6}}

\Crefformat{equation}{#2Equation~\textup{(#1)}#3}
\Crefrangeformat{equation}{Equations~\textup{#3(#1)#4--#5(#2)#6}}
\Crefmultiformat{equation}{Equations~\textup{#2(#1)#3}}{ and \textup{#2(#1)#3}}
{, \textup{#2(#1)#3}}{, and \textup{#2(#1)#3}}
\Crefrangemultiformat{equation}{Equations~\textup{#3(#1)#4--#5(#2)#6}}%
{ and \textup{#3(#1)#4--#5(#2)#6}}{, \textup{#3(#1)#4--#5(#2)#6}}{, and \textup{#3(#1)#4--#5(#2)#6}}

\crefdefaultlabelformat{#2\textup{#1}#3}

\title{\LARGE \bf
Nonlinear input design as optimal control of a Hamiltonian system
}

\author{Jack Umenberger and Thomas B. Sch\"on
\thanks{J. Umenberger and T.B. Sch\"on are with the Dept. of Information Technology, Uppsala University, 751 05 Uppsala, Sweden.  
E-mails: \{jack.umenberger, thomas.schon\}@it.uu.se.}}%

\begin{document}

\maketitle
\thispagestyle{empty}
\pagestyle{empty}

\begin{abstract}
We propose an input design method for a general class of parametric probabilistic models, 
including nonlinear dynamical systems with process noise.
The goal of the procedure is to select inputs such that the parameter posterior distribution concentrates about the true value of the parameters; however, exact computation of the posterior is intractable.
By representing (samples from) the posterior as trajectories from a certain Hamiltonian system, we transform the input design task into an optimal control problem.
The method is illustrated via numerical examples, including
MRI
pulse sequence design.
\end{abstract}


\section{Introduction}
Building mathematical models of systems from measured input-output data is a task of central importance in a number of fields, from science and engineering, to medicine and finance.
Well-chosen inputs can improve the quality of the model by exciting the system in such a way as to extract relevant information.
The process of choosing such inputs, subject to experimental constraints, is known as \emph{input design}.

The information content associated with an input is often quantified by (some scalar function of) the Fisher Information Matrix (FIM) \cite[\mysec7.4]{Ljung1999}. 
For models in which the mapping from inputs to outputs is linear, the FIM is an affine function of the input spectrum, and so input design in the frequency domain
can be formulated as a convex optimization problem.
For general nonlinear models, even computation of the FIM can be a challenging task that must be carried out approximately, making optimal input design difficult; 
cf. \emph{related work} below for further details.

In this paper, we take a different approach to input design for general nonlinear models, and optimize the posterior distribution over model parameters directly.
Using tools from statistical inference, namely Hamiltonian Monte Carlo (HMC) \cite{duane1987hybrid}, we construct a Hamiltonian system, trajectories of which correspond to samples from the posterior distribution.
The input design problem of shaping the posterior for accurate parameter inference can then be formulated as an optimal control problem for this Hamiltonian system.

\subsubsection{Related work}\label{sec:related_work}
While the approach presented in this paper is applicable to any model satisfying the assumptions stated in \mysec\ref{sec:problem}, our focus shall be on input design for \emph{dynamical} systems.
Input design for linear systems is by now a mature topic, with a large body of literature, cf. e.g. \cite{goodwin1977dynamic,Ljung1999,jansson2005input,manchester2010input} and the references therein.
In contrast, results on input design for nonlinear systems are more scarce.
Input design for structured systems (i.e. the interconnection of linear dynamics with a static nonlinearity) was studied in \cite{vincent2010input,gevers2012experiment} and \cite{mahata2016information}, with the latter proving that D-optimal inputs can be realized by a mixture of Gaussians for certain Wiener systems. 
Ideas from input design for impulse response systems were extended to nonlinear systems admitting a nonparametric Volterra series representation in \cite{birpoutsoukis2018}.
Similarly, design for nonlinear FIR models was considered in \cite{hjalmarsson2007optimal,larsson2010optimal,de2016d}.
We note that for linearly parametrized static nonlinear systems, input design is a convex problem; this observation motivated the work
\cite{forgione2014experiment} which considers inputs restricted to a finite number of amplitude levels, for dynamical systems.

\subsubsection{Paper structure}\label{sec:structure}
The paper proceeds as follows: 
a detailed problem specification is provided in \mysec\ref{sec:problem}, followed by a brief primer on HMC in \mysec\ref{sec:hmc}.
The Hamiltonian optimal control problem is presented in \mysec\ref{sec:ideal_problem}, while approximations required to solve the problem, along with the final proposed algorithm, are stated in \mysec\ref{sec:approximations}. 
Numerical experiments illustrating the method are presented in \mysec\ref{sec:examples}.

\section{Problem set-up}\label{sec:problem}
In this work, the object of interest is a system, $\sys$, that accepts inputs $u$ and generates measurable outputs $y$.
The relationship between inputs and outputs is modeled by a probability distribution 
$p(y,\param|u)=p(y|u,\param)p(\param)$,
where $\param\in\real^{\nparam}$ denotes the model parameters. 
Here $p(y|u,\param)$ is the likelihood, and $p(\param)$ denotes any prior belief we may have over the model parameters.

\begin{assumption}\label{ass:model}
There exists $\paramtr\in\real^{\nparam}$ such that the data-generating distribution for $\sys$ is given by $p(y|u,\paramtr)$.
\end{assumption}
That is, we assume the existence of `true parameters' $\paramtr$, such that, given $u$, the output $y$ from $\sys$ is distributed according to $p(y|u,\paramtr)$.
Our goal is to design an input $u^*$ such that $\paramtr$ can be accurately inferred from the input/output data $(u^*,y^*)$.
In this paper, we shall primarily be concerned with \emph{dynamical} systems, $\sys$, for which the inputs and outputs shall be time series, 
i.e., 
$u=\lbrace u_t\rbrace_{t=1}^\numdp$ and 
$y=\lbrace y_t\rbrace_{t=1}^\numdp$. 
However, the approach applies to any probabilistic input/output system, e.g., systems for which the mapping from $u$ to $y$ is static.

The main assumption governing the class of probabilistic systems to which our approach is applicable is as follows:
\begin{assumption}\label{ass:latent}
There exist latent variables $x$ such that gradients of the log joint-likelihood, $\log p(y,x|u,\param)$, w.r.t. $\param$ and $x$, can be computed. 
Furthermore, any prior distribution over $\param$ is chosen such that $\nabla_\theta\log p(\param)$ can be computed. 
\end{assumption}  
Assumption \ref{ass:latent} is necessary for the application of HMC; cf., \mysec\ref{sec:hmc}.
In this paper, we will work with state-space models for which the joint-likelihood 
$p(x,y|u,\param)$
decomposes as
\begin{align*}
p(x_\inind|\param){\prod}_{t=1}^{\numdp}p(y_t|x_t,u_t,\param)p(x_{t}|x_{t-1},u_{t-1},\param),
\end{align*}
where the internal states $x=\lbrace x_t\rbrace_{t=1}^\numdp$ constitute suitable latent variables for Assumption \ref{ass:latent}.
As a concrete example, consider the nonlinear dynamical system
\begin{align}
x_{t+1} &= f_\param(x_t,u_t) + w_t, \ y_t = g_\param(x_t,u_t) + v_t, \label{eq:nonlinear_state_space} \\ 
\ w_t&\sim\normal{0}{\Sigma_{w,\param}},  \ v_t\sim\normal{0}{\Sigma_{v,\param}}, \ x_\inind \sim \normal{\mu_\param}{\Sigma_{\inind,\param}}, \nonumber
\end{align}
For this model, 
$\log p(y_t|x_t,u_t,\param)\propto |y_t-g_\param(x_t,u_t)|_{\Sigma_{v,\param}^{-1}}^2$
and 
$\log p(x_{t}|x_{t-1},u_{t-1},\param) \propto |x_t-f_\param(x_{t-1},u_{t-1})|_{\Sigma_{w,\param}^{-1}}^2$, 
so $\log p(x,y|u,\param) $ is differentiable, when $f_\param$, $g_\param$, $\Sigma_{w,\param}$, $\Sigma_{v,\param}$ and $\Sigma_{\inind,\param}$ are differentiable.
For the special-case of linear, Gaussian dynamical systems this assumption holds with $x=\varnothing$ (no latent variables). 

To quantify the notion of `accurately inferring' $\paramtr$, 
we seek inputs $u$ that cause the posterior $p(\param|u,y)$ to concentrate about the true parameters, i.e., we seek to minimize 
\begin{equation}\label{eq:cost}
\pstcost(u,y) \defeq \evpost\left[|\param-\paramtr|^2\right] = \int |\param-\paramtr|^2p(\param|u,y)d\param.
\end{equation}
Of course, we do not have access to outputs $y$ before providing an input $u$, so we shall instead seek to minimize 
\begin{equation}\label{eq:ecost}
\epstcost(u) \defeq 
\evdata\left[\pstcost(u,y)\right] = \int \pstcost(u,y) p(y|u,\paramtr)dy,
\end{equation}
i.e., the \emph{expected} cost, where the expectation is w.r.t. the true data-generating distribution, $p(y|u,\paramtr)$, given $u$. 
Both \eqref{eq:cost} and \eqref{eq:ecost} depend on the true parameters, $\paramtr$.

\begin{assumption}\label{ass:true}
The `true' parameter value $\paramtr$ is known.
\end{assumption}
This is a somewhat unrealistic, albeit ubiquitous, assumption in the experiment design literature \cite[\mysec13]{Ljung1999}, as $\paramtr$ is the very quantity we wish to estimate. 
Nevertheless, in applications some reasonable guess for $\paramtr$ can usually be made; alternatively, multi-stage adaptive or robust optimization \cite{rojas2007robust} may be employed to iteratively improve the initial estimate.

Finally, our choice of inputs is subject to constraints.
Let $\feasinput$ denote the set of feasible inputs; 
common choices include constraints on input amplitude, 
e.g., $\feasinput = \lbrace u: \umin \leq u_t \leq \umax, \ \forall t \rbrace$, 
and/or power, 
e.g., $\feasinput = \lbrace u: \sum_t u_t'u_t \leq \maxpower \rbrace$ for some $\maxpower\in\real$.
The only requirement for $\feasinput$ is that projection onto this set can be carried out efficiently.

To summarize, the goal of the paper is to solve the input design problem:
$\min_{u\in\feasinput} \epstcost(u)$.
The principal difficulty is computing expected values w.r.t. the posterior distribution $p(\param|u,y)$ for the general class of probabilistic systems considered in this paper.
Our key to circumventing this difficulty is the construction of a Hamiltonian system, the trajectories of which correspond to samples from the posterior.
This system forms the basis of HMC, which we describe next.

\section{Hamiltonian Monte Carlo}\label{sec:hmc}
In this section we provide a brief pragmatic introduction to the use of HMC for statistical inference; for further details, cf. \cite{neal2011mcmc}.
Consider a physical system, 
denoted $\hamsys$,
described by position $\hpos\in\real^{\nq}$ and momentum $\hmom\in\real^\nq$, which collectively form the system state $(\hpos,\hmom)$.
Let $K(\hmom)$ and $U(\hpos)$ denote the kinetic and potential energy of this system, respectively.
The behavior of 
$\hamsys$ 
is then characterized by its \emph{Hamiltonian}, $H(\hpos,\hmom)=U(\hpos)+K(\hmom)$, and Hamilton's equations of motion:
\begin{equation}\label{eq:hamiltons_equations}
\frac{d\hpos_i}{dt} = \frac{\p H}{\p \hmom_i}, \quad \frac{d\hmom_i}{dt} = -\frac{\p H}{\p \hpos_i}, \quad \textup{for } i=1,\dots,d.
\end{equation}
The Hamiltonian defines a canonical distribution
(in the statistical mechanics sense) for the states $(\hpos,\hmom)$, defined by
\begin{equation}\label{eq:canonical}
\canonical_{u,y}(\hpos,\hmom) = (1/Z)\exp\left(-(H(\hpos,\hmom))/T\right), 
\end{equation} 
where $T$ is the temperature of the system, and $Z$ is an appropriate normalization constant.
From the perspective of statistical mechanics, a given energy function $H$ leads to a probability distribution $\canonical_{u,y}$.
For the purpose of statistical \emph{inference}, HMC takes the opposite perspective: given a probability distribution of interest (a `target distribution'), $\target(\param)$, we can define a Hamiltonian $H$ such that the canonical distribution in \eqref{eq:canonical} coincides with $\target(\param)$.
In particular, one can let the positions represent the parameters of interest, i.e. $\hpos=\param$, and choose $U(\hpos)=-\log\target(\hpos)$.
The momentum variables enter only for the purpose of constructing the Hamiltonian system; we are not interested in the distribution over $\hmom$, for the purpose of inference.
Given this, a popular choice for the kinetic energy is the quadratic function $K(\hmom)=\hmom'\Inertia^{-1}\hmom$, which leads to a Gaussian distribution over $\hmom$.
Here, $\Inertia=mI$ can be interpreted as the (positive definite) inertia matrix.
{Whenever HMC is used in this paper, the target distribution is $\target(\hpos)=\posterior(\param|u,y)$, i.e., the posterior distribution over the parameters $\param$ of $\sys$, given $u$ and $y$.}

To generate samples from $\target(\hpos)$, a Markov chain is constructed as follows:
given a current state $(\hpos,\hmom)$, a new state $(\hpos^*,\hmom^*)$ is proposed by first re-sampling the momentum from its canonical (Gaussian) distribution, and then forward-simulating the Hamiltonian dynamics \eqref{eq:hamiltons_equations} for some finite time, $\simlength$.
To simulate the Hamiltonian dynamics it is necessary to employ a discrete-time approximation of \eqref{eq:hamiltons_equations}.
A popular choice is the so-called `leapfrog' method (which is reversible and preserves volume exactly), one iteration $(t\rightarrow t+\stepsize)$ of which involves the following sequence of steps:
\begin{subequations}\label{eq:leapfrog}
	\begin{align}
	\hmom_i(t+\stepsize/2) &= \hmom_i(t) - (\stepsize/2)
\nabla_{\hpos_i}{U(\hpos(t))}, \label{eq:leapfrog_one} \\
	\ \hpos_i(t+\stepsize) &=\hpos_i(t) + (\stepsize/m)\hmom_i(t+\stepsize/2),  \\
	\ \hmom_i(t+\stepsize) &= \hmom_i(t+\stepsize/2) - (\stepsize/2)\nabla_{\hpos_i} U(\hpos(t+\stepsize)).
	\end{align}
\end{subequations}
The proposed state $\hpos^*$  is then accepted (as a new sample from $\target(\hpos)$) with probability 
$\min(1,\exp(-H(\hpos^*,\hmom^*)+H(\hpos,\hmom)))$.
Note that the accept/reject step is only required to correct for errors in the numerical integration of \eqref{eq:hamiltons_equations}.
To run HMC, the following parameters $\hmcparam=\lbrace m,\simlength,\stepsize\rbrace$ must be specified; see e.g. \cite{neal2011mcmc} for practical guidelines on these choices.

\section{Input design via Hamiltonian dynamics}\label{sec:ideal_problem}
Equipped with this knowledge of HMC, let us first present the intuition for the proposed approach to input design.

\subsection{Intuitive explanation and illustration of the approach}\label{sec:illustration}

Consider a probabilistic system $\sys$ and the associated Hamiltonian system $\hamsys$, as defined in \mysec\ref{sec:problem} and \mysec\ref{sec:hmc}, respectively.
Given inputs/outputs $(u,y)$ from $\sys$, trajectories of $\hamsys$ correspond to samples from the posterior $p(\param|u,y)$.
The cost $\pstcost(u,y)$ can be thought of, roughly speaking, as the expected distance of a trajectory of $\hamsys$, given initial conditions sampled from $\canonical_{u,y}$.
To minimize $\pstcost(u,y)$ we can think of solving an optimal control problem for $\hamsys$:
choose $u$ (and, indirectly, $y$) so as to drive the position component ($q=\param$) of the state, after some finite time, close to the target state, $\qnom=\paramtr$.
This idea is illustrated in Fig. \ref{fig:illustrate}, by application to a linear dynamical system
\begin{align}
x_{t+1}=Ax_t+Bu_t+w_t, \ y_t=Cx_t+Du_t+v_t.
\end{align}
with $w_t\sim\normal{0}{\Sigma_w}$ and $v_t\sim\normal{0}{\Sigma_v}$.
Let $\param=\lbrace A_{21}, A_{22}\rbrace$, i.e., these are the unknown parameters that we wish to estimate from data.
True parameter values are given by $\paramtr=\lbrace -0.2,0.7\rbrace$.
We shall return to this example in \mysec\ref{sec:ex_linear}, where further details are provided, but for now focus on the two inputs shown in Fig. \ref{fig:illustrate}(a): 
`nominal' $\unom$ and `optimized' $\uopt$.
Trajectories of the Hamiltonian systems corresponding to these inputs are shown in Fig. \ref{fig:illustrate}(b).
Though initialized at the same positions, the (final value of) trajectories corresponding to $\uopt$ tend to concentrate around $\qnom$, more so than those corresponding to $\unom$.
In Fig. \ref{fig:illustrate}(c) we plot the posterior distributions (computed exactly, using a Kalman filter), and observe that $p(\param|\uopt,\yopt)$ is indeed much more `peaked' around $\paramtr$ compared to $p(\param|\unom,\ynom)$.

\begin{figure}
	\centering
	\subfloat[Different input sequences to the linear dynamical system.]{\includegraphics[width=\columnwidth]{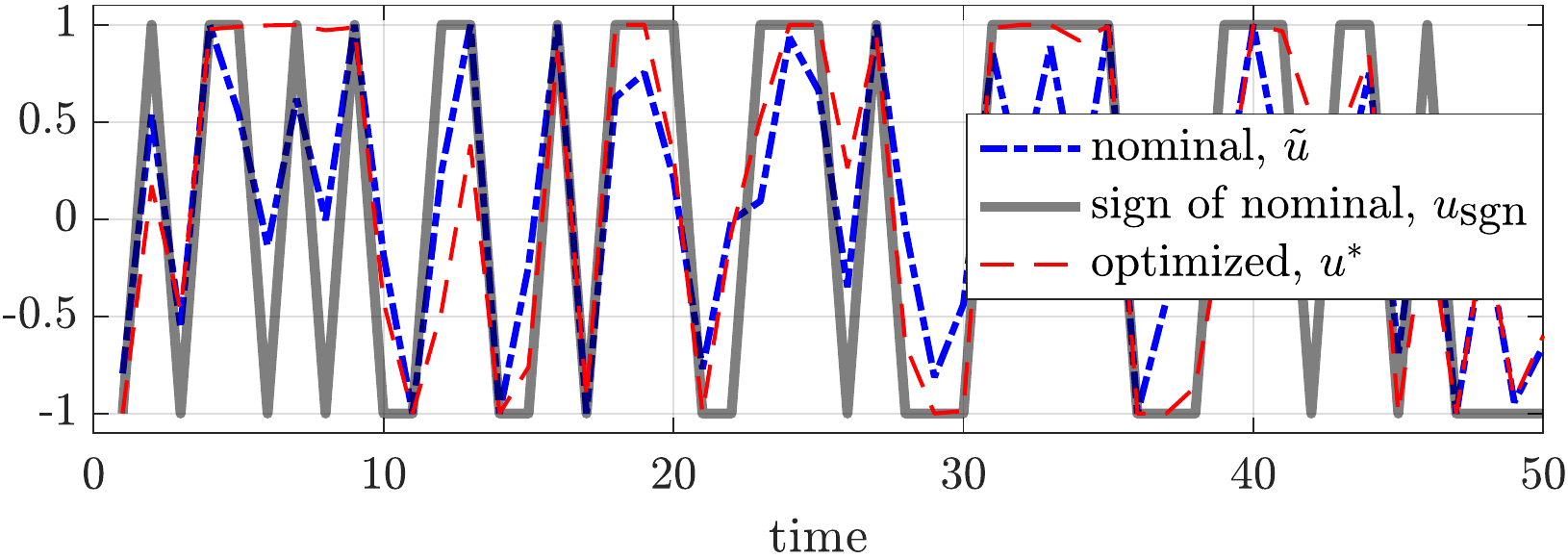}}\\
	\subfloat[Position component of trajectories of the Hamiltonian systems associated with the different inputs in (a). 
	The squared error $\sum|\hpos(\simsteps)-\qnom|^2$ is 1.34 and 0.39 for $\unom$ and $\uopt$, respectively.
	15 trajectories per input are depicted.
	]{\includegraphics[width=\columnwidth]{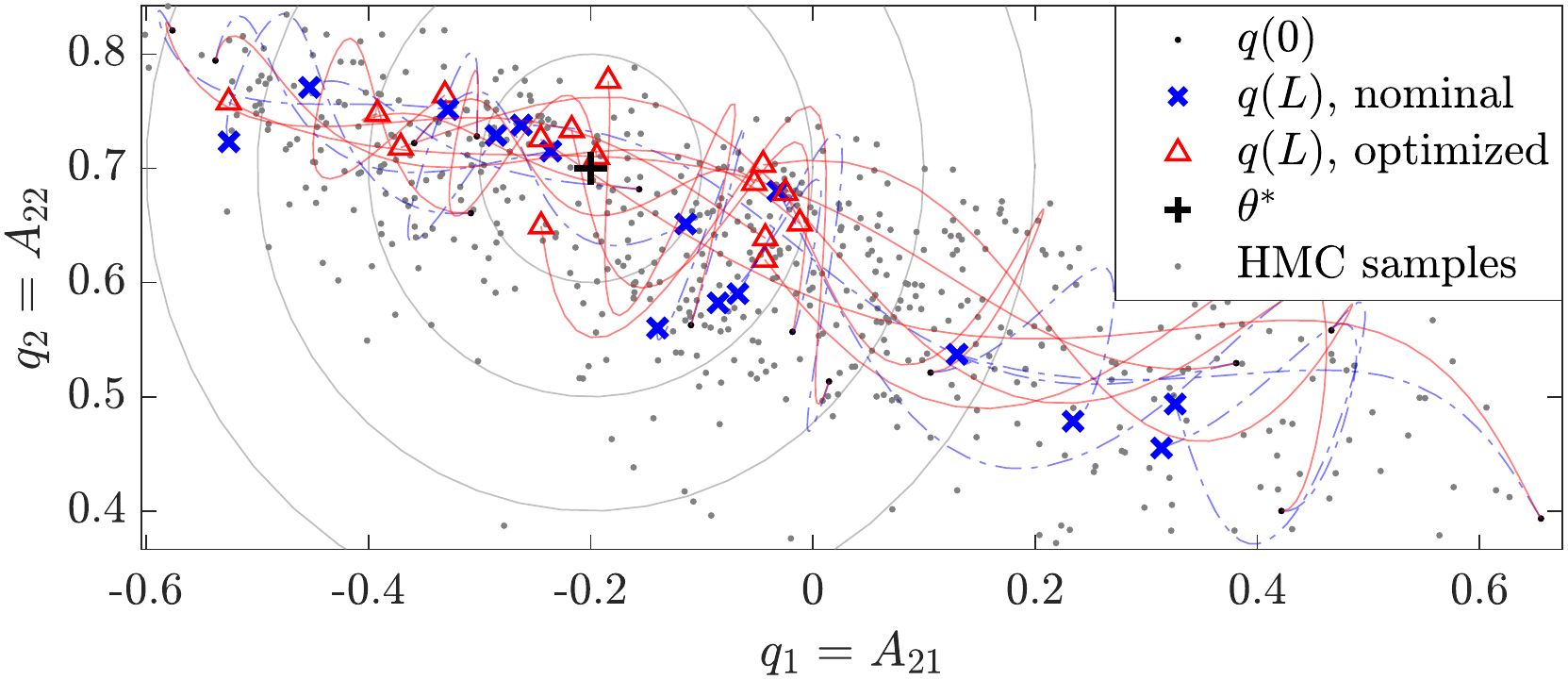}} \\
	\subfloat[Posterior distributions corresponding to inputs in (a).]{\includegraphics[width=\columnwidth]{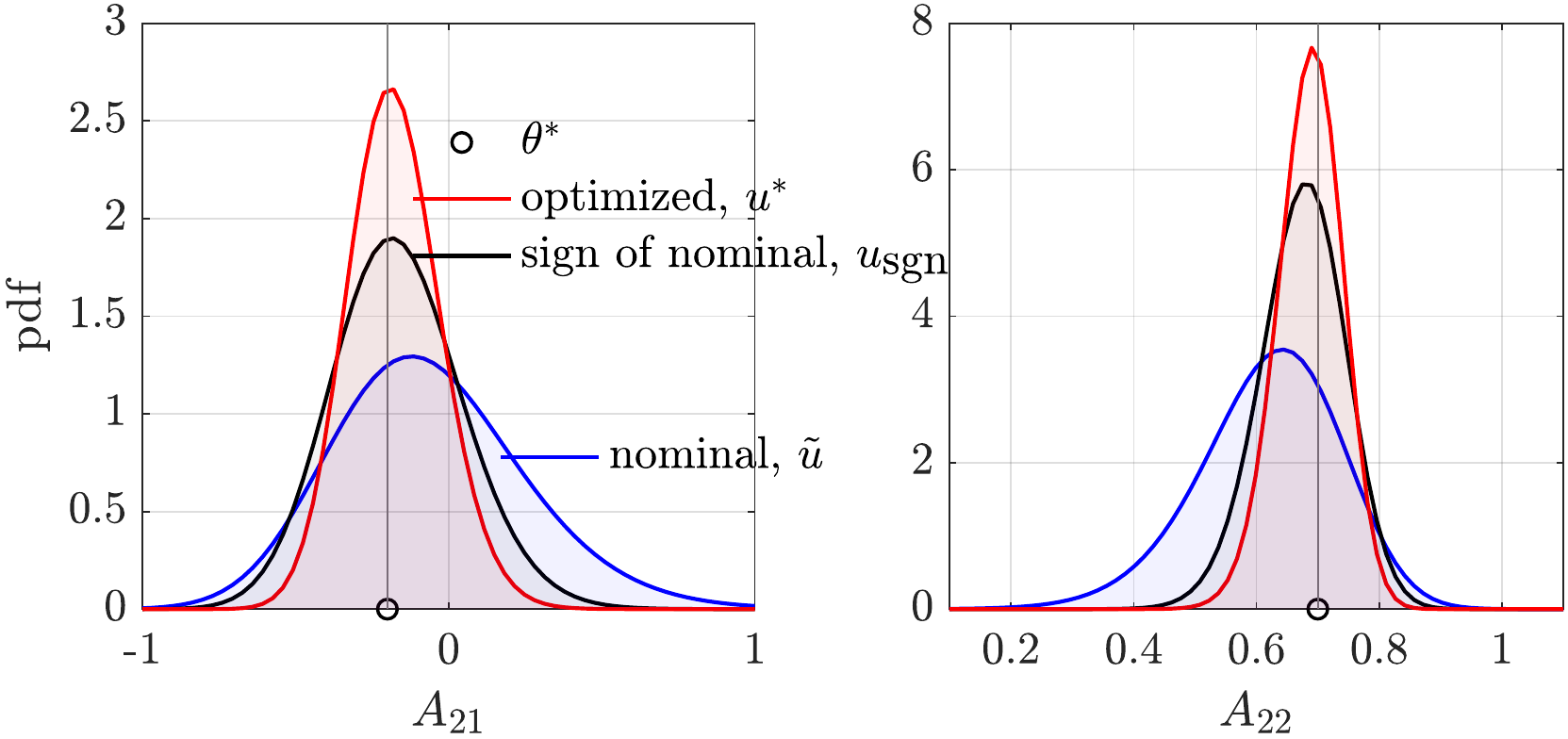}} 
	\caption{Illustration of proposed method on a linear system, cf. \mysec\ref{sec:illustration} for details.}
	\label{fig:illustrate}
\end{figure}

\subsection{Computing the cost function with Hamiltonian dynamics}

To make this idea more precise, we now relate the cost function $\pstcost$ defined in \eqref{eq:cost} to the optimal control problem described above.
Let $\qmap_\simlength(u,y,\qinit,\pinit)$ denote the new state obtained by (exactly) evolving the Hamiltonian system $\hamsys$ (given $u$ and $y$) forward for duration $\simlength$, starting from the initial state $(\qinit,\pinit)$.
For ease of exposition, suppose that $\nabla_\param \log p(y|u,\param)$ can be evaluated, such that the position of the Hamiltonian can be chosen as $\hpos=\param$.
We will relax this assumption in the sequel. 
Similarly, let $\qnom=\paramtr$.
Then, the optimal control problem described above can be formulated as 
\begin{equation}\label{eq:opt_control_cost}
\min_{u\in\feasinput} \ \evcanon\left[ |\qmap_\simlength(u,y,\qinit,\pinit) - \qnom|^2 \right],
\end{equation}
where $\evcanon$ denotes expectation w.r.t. the canonical distribution $\canonical_{u,y}$ in \eqref{eq:canonical}, given $u$ and $y$.
We then have the following equivalence between \eqref{eq:opt_control_cost} and $\min_{u\in\feasinput} \ \pstcost(u,y)$:

\begin{lemma}\label{lem:cost_equiv}
Consider $\pstcost(u,y)$ defined in \eqref{eq:cost}. For all $\simlength$,
	\begin{align*}
	\pstcost(u,y) = \evcanon \left[ |\qmap_\simlength(u,y,\qinit,\pinit) - \qnom|^2 \right]. 
	\end{align*}
\end{lemma}

\begin{proof}
As the Hamiltonian dynamics \eqref{eq:hamiltons_equations} leave the canonical distribution $\canonical_{u,y}$ invariant \cite{neal2011mcmc}, 
$\qmap_\simlength(u,y,\qinit,\pinit)$
is distributed according to $\posterior(\param|u,y)$.
Furthermore, as $q=\param$ and $\qnom=\paramtr$, we have
$\ev_{\qinit,\pinit} \left[ |\qmap_\simlength(u,y,\qinit,\pinit) - \qnom|^2 \right] = \evpost \left[|\param-\paramtr|^2\right]$. 
\end{proof}

It is worth noting that the usual `bias-variance decomposition' applies to $\pstcost(u,y)$.
Recall that $\evpost$ denotes expectation w.r.t. $\posterior(\param|u,y)$. 
With $\pstmean\defeq\evpost\left[\param]\right]$ we have
\begin{subequations}
	\begin{align*}
	&\evpost \left[|\param-\paramtr|^2\right] 
	= \evpost \left[|\param-\pstmean + \pstmean - \paramtr|^2 \right] \\
	= & \evpost \left[ |\pstmean - \paramtr|^2 + 2(\param-\pstmean)'(\pstmean - \paramtr) + |\param-\pstmean|^2 \right]  \\
	= & |\pstmean - \paramtr|^2 + \trace \evpost \left[(\param-\pstmean)(\param-\pstmean)'\right],
	\label{eq:cancel}
	\end{align*}
\end{subequations}
i.e., minimizing $\pstcost(u,y)$ (or solving \eqref{eq:opt_control_cost}) is equivalent to minimizing the sum of
the 2-norm error between the posterior mean and the true parameters,
and the trace of the posterior variance.
This can be observed in Fig \ref{fig:illustrate}(c).

For ease of exposition, we assumed that $\nabla_\param \log p(y|u,\param)$ could be computed. 
For general probabilistic systems, e.g., nonlinear state-space models such as \eqref{eq:nonlinear_state_space}, this is not the case.
However, as stated in Assumption \ref{ass:latent}, we only require that gradients of $\log p(y,x|u,\param)$ are computable, for some latent variable $x$, cf., \mysec\ref{sec:problem}.
In this general setting, we augment the state of $\hamsys$ with these latent variables, and choose $\hpos=\left[\param;x\right]$.
From \eqref{eq:leapfrog_one}, HMC requires $\nabla_\hpos U(\hpos)$, where
\begin{equation*}
-U(\hpos) = \log p(q|u,y) \propto \log\left( p(y|u,\hpos)p(\hpos|u)\right).
\end{equation*}
For $\hpos=\left[\param;x\right]$, $\log\left( p(y|u,\hpos)p(\hpos|u)\right)$ decomposes as
\begin{align*}
&\log p(y|u,x,\param) + \log p(x|u,\param) + \log p(\param) \\
= &\log p(y,x|u,\param) + \log p(\param),
\end{align*}
hence Assumption \ref{ass:latent} is sufficient for computation of $\nabla_q U(q)$.
We can then optimize a weighted version of \eqref{eq:opt_control_cost}, $\evcanon\left[ |\qmap_\simlength(u,y,\qinit,\pinit) - \qnom|_\weight^2 \right]$,
in which $\qnom=[\paramtr;0]$ and $\weight=\blkdiag(I,0)$.
This weighted version ignores the components of $\hpos$ associated with $x$, so that we retain
$\ev_{\qinit,\pinit} \left[ |\qmap_\simlength(u,y,\qinit,\pinit) - \qnom|^2 \right] = \evpost \left[|\param-\paramtr|^2\right]$. 

To summarize, Lemma \ref{lem:cost_equiv} established equivalence between the optimal control problem \eqref{eq:opt_control_cost} and the cost $\pstcost(u,y)$ in \eqref{eq:cost}.
As we wish to optimize the expected cost, $\epstcost(u)$ given in \eqref{eq:ecost}, the optimization problem we propose solving for input design is given by
\begin{equation}\label{eq:proposed_opt_prob}
\uopt = \arg\max_{u\in\feasinput} \ \evdata\left[ \evcanon \left[ |\qmap_\simlength(u,y,\qinit,\pinit) - \qnom|^2 \right] \right],
\end{equation}
where $ \evdata$ denotes expectation w.r.t. $p(y|u,\paramtr)$.

\section{Approximate solution to control problem}\label{sec:approximations}
The optimization problem \eqref{eq:proposed_opt_prob}, as stated, is difficult to solve. 
In this section we propose a number of approximations to obtain a tractable alternative.
At the outset, we emphasize that we are only interested in local optimization:
i.e., we will assume some nominal input $\unom$, and seek an improved input by numerical optimization.

\subsection{Approximating  expected values}
There are two expected values in \eqref{eq:proposed_opt_prob}, neither of which can be computed exactly for the general probabilistic models under consideration. 
First, we consider approximation of $\evdata$, i.e., expectation w.r.t. the generative distribution $p(y|u,\paramtr)$.
One approach would be to simply approximate it with a Monte Carlo (MC) average, using samples $y^i$ drawn from $y^i\sim p(y|\unom,\paramtr)$.
However, in this approach, the dependence of $y$ on $u$ is lost during optimization of $u$, because $y$ is generated with a fixed, nominal $\unom$.
As an alternative, we consider approximate realizations $\ydet^i(u)$ from $y^i\sim p(y|u,\paramtr)$, where $\ydet^i(u)$ is a deterministic function of~$u$.
The idea is to capture (some of) the dependence of~$y$ on~$u$.
For example, consider the state-space model in~\eqref{eq:nonlinear_state_space}.
The approximate realization~$\ydet^i(u)$ can be formed by first generating samples $(w^i_{1:\numdp}, v^i_{1:\numdp}, x_0^i)$ of the stochastic processes in \eqref{eq:nonlinear_state_space}, and then setting
$\xdet_{0}^i(u)=x_0^i$ and
\begin{equation*}
\ydet_t^i(u) = g_\paramtr(\xdet_t^i(u),u_t) + v^i_t, \ \xdet_{t+1}^i(u) = f_\paramtr(\xdet_{t}^i(u),u_t) + w_t^i.
\end{equation*}

We then approximate \eqref{eq:proposed_opt_prob} with
\begin{equation}\label{eq:proposed_approx_y}
\frac{1}{\numy} \sum_{i=1}^\numy \evcanon \left[ |\qmap_\simlength(u,\ydet^i(u),\qinit,\pinit) - \qnom|^2 \right].
\end{equation}
Note that in practice we always choose $\numy=1$, i.e., we use a single deterministic approximation $\ydet(u)$ formed by setting the stochastic processes to their mean values, e.g., for the state-space model \eqref{eq:nonlinear_state_space} the values of $(w_{1:\numdp}, v_{1:\numdp}, x_0)$ would each be set to zero.

Next, we turn our attention to the approximation of $\evcanon$, i.e., expectation w.r.t. the canonical distribution $\canonical_{u,y}$ in~\eqref{eq:canonical}.
As $\canonical_{u,y}$ depends on $(u,y)$,
the natural MC approximation of \eqref{eq:proposed_approx_y} (for $\numy=1$) would involve samples $(\qinit^i,\pinit^i)$ drawn from $\canonical$ defined by $(u,\ydet(u))$.
For simplicity, we ignore this dependency, and instead use a MC approximation with samples drawn from $\canonical_{\unom,\ydet(\unom)}$, 
 i.e., the canonical distribution defined by the nominal input $\unom$ and the approximate output $\ydet(\unom)$.
Such samples can be generated by running HMC.
Thus, with $(\qinit^i,\pinit^i)\sim \canonical_{\unom,\ydet(\unom)}(\hpos,\hmom)$, we approximate \eqref{eq:proposed_opt_prob} (with $\numy=1$) by
\begin{equation}\label{eq:proposed_opt_mc}
\frac{1}{\numq} \sum_{i=1}^\numq |\qmap_\simlength(u,\ydet(u),\qinit^i,\pinit^i) - \qnom|^2.
\end{equation}

\subsection{Discrete-time approximation of dynamics}

The cost in~\eqref{eq:proposed_opt_mc} depends on the mapping $\qmap$, which involves the evolution of the continuous-time dynamical system $\hamsys$, cf. \eqref{eq:hamiltons_equations}.
For the purpose of optimization, we replace the dynamics in \eqref{eq:hamiltons_equations} with the discrete-time approximation in \eqref{eq:leapfrog}.
Then minimization of \eqref{eq:proposed_opt_mc} w.r.t. $u$ can be approximated by the following discrete-time optimal control problem:
\begin{align}\label{eq:dt_opt_prob}
\min_{u\in\feasinput} \ &\frac{1}{\numq} \sum_{i=1}^\numq |\hpos^i(L) - \qnom|^2, \\
\st \ \ & \hpos^i(0) = \ \qinit^i, \ \hmom^i(0) = \pinit^i,  \ i=1,\dots,\numq \nonumber \\
& \hmom^i(1) = \hmom^i(0) + \gradlog(\hpos^i(0)), \ i=1,\dots,\numq \nonumber \\ 
& \hpos^i(k) = \hpos^i(k-1) + (\stepsize/m)\hmom^i(k), \ k=1,\dots,L, \nonumber \\
& \hmom^i(k) = \hmom^i(k-1) + \gradlog(\hpos^i(k-1)), \ k=1,\dots,L, \nonumber
\end{align}
where $h(\hpos) = (\stepsize/2)\nabla_\hpos \log\left(p(\ydet(u)|u,\hpos)p(\hpos)\right)$,
and $\simsteps = \simlength/\stepsize$.
We then approximately solve \eqref{eq:dt_opt_prob}, i.e. locally optimize it, using standard numerical optimization methods \cite{diehl2009efficient}, such as projected gradient descent.

\begin{algorithm}
	\KwData{probabilistic model, $\paramtr$, $\feasinput$, $\unom\in\feasinput$, $\numq$, $\utol$}
	Initialize $u^{(0)}\leftarrow\unom$ and $\loopi\leftarrow0$\;
	\While{ $\Vert u^{(\loopi)} - u^{(\loopi-1)}\Vert_1 \geq \utol$ }{
	Generate approximate output $\ydet(u^{(k)})$\;
	Run HMC to generate $\numq$ samples from \eqref{eq:canonical}, i.e.,
	$\lbrace \qinit^i,\pinit^i\rbrace_{i=1}^\numq\sim \canonical_{u^{(\loopi)},\ydet(u^{(\loopi)})}(\hpos,\hmom)$\; \label{alg:hmc}
	Solve \eqref{eq:dt_opt_prob}, locally, using projected gradient descent, initialized with $u^{(\loopi)}$. Let $\utmp$ denote the solution\;
	$\loopi\leftarrow \loopi+1$ and $u^{(\loopi)} \leftarrow \utmp$\; \label{alg:solve}
	}
	\Return $u^{(\loopi)}$
	\caption{proposed approach to input design}\label{alg:inputdesign}
\end{algorithm}

\subsection{Discussion}
We conclude this section with a brief discussion of some extensions, limitations, and practical considerations. 

\discuss{\emph{Sensitivity to $\paramtr$} \quad The choice of $\paramtr$ (i.e. the best guess of the unknown true parameter) determines the point about which we attempt to concentrate the posterior, $\qnom$, as well as deterministic output approximation, $\ydet$. It is possible to reduce the sensitivity of the former to errors in our guess of $\paramtr$ by using a `dead-zone' loss function in \eqref{eq:cost}, instead of the 2-norm \cite[\mysec6.1.2]{Boyd2004}. Note that $\paramtr$ does not affect HMC.}

\discuss{\emph{Choice of HMC parameters} \quad The method requires the user to specify $\hmcparam$, cf. \mysec\ref{sec:hmc}, 
which may be tuned via pilot runs (of HMC).
For instance, the stepsize $\stepsize$ can be increased until the acceptance rate falls below a reasonable tolerance; cf. \cite{neal2011mcmc} for guidelines. 
}


\discuss{\emph{Computational complexity} \quad The key factors that influence computational complexity are: i) running HMC, ii) solving \eqref{eq:dt_opt_prob}, and iii) the number of samples $\numq$.
Note that~$\numq$ directly affects both i) and ii). 
The larger~$\numq$, the more iterations of HMC are required to generate the samples. 
$\numq$ is typically on the order of 10-100, which means that the cost of solving \eqref{eq:dt_opt_prob} dominates that of running HMC.
The cost of \eqref{eq:dt_opt_prob} scales linearly with $\numq$, as it is equivalent to solving $\numq$ simultaneous nonlinear model predictive control (MPC) problems.
To reduce the computational burden, it is possible to employ stochastic gradient methods, and use a subset of $\lbrace \qinit^i,\pinit^i\rbrace_{i=1}^\numq$ to compute the gradient at each iteration.	
}

\section{Numerical examples}\label{sec:examples}

\subsection{Illustration on a linear system}\label{sec:ex_linear}
We now return to the illustration presented in \mysec\ref{sec:illustration}, and fill-out the remaining details in light of \mysec\ref{sec:approximations}.
As we can compute $\nabla_\param \log p(y|u,\param)$ exactly, the position of the Hamiltonian system $\hamsys$ can be set to $q=\param\in\real^2$, i.e., for this system, there is no need to include latent variables $x$ in the Hamiltonian state.
We randomly generate $\unom_{1:\numdp}$ by sampling from $\normal{0}{1}$ for $\numdp=50$ and project it onto the feasible set $\feasinput=\lbrace u: |u_t|\leq 1\rbrace$, cf. Fig. \ref{fig:illustrate}(a).
HMC is run (line~\ref{alg:hmc} of \myalg~\ref{alg:inputdesign}) to generate samples from $p(\param|\unom,\ydet(\unom))$, cf. Fig.~\ref{fig:illustrate}(b).
Again, cf. \cite{umenberger19hmcsoftware} for HMC parameters.
The optimization problem \eqref{eq:dt_opt_prob} is then solved, locally with projected gradient descent as in line \ref{alg:solve}, to obtain $\uopt$, which is shown in Fig. \ref{fig:illustrate}(a).
Trajectories of the Hamiltonian systems corresponding to each of these inputs, $\unom$ and $\uopt$, are shown in Fig. \ref{fig:illustrate}(b).
Notice that the final positions of the trajectories corresponding to the optimized input tend to concentrate about $\qnom=\paramtr$.
This is reflected in the posterior distribution $p(\param|\uopt,\ydet(\uopt))$, plotted in Fig. \ref{fig:illustrate}(c), which is more peaked around $\paramtr$.
This is perhaps unsurprising, given that the amplitude of $\uopt$ is typically larger than $\unom$. 
For this reason, we also plot the posterior for $\usgn=\sign(\unom)$ to demonstrate that the improvement observed after optimization is due to more than just a simple increase in the amplitude of the input. 

\subsection{Nonlinear state-space system}\label{sec:ex_nl}
In this section, we consider input design for the following nonlinear state-space model
\begin{align*}
x_{t+1} = \param\frac{x_t^3-x_t}{1+x_t^2} + u_t +w_t, \ y_t = x_t+x_t^2+\frac{1}{10}x_t^3+v_t,
\end{align*}
where $w_t,v_t,x_1$ are normally distributed according to $\normal{0}{0.1^2}$.
The goal is to design the input $u_{1:\numdp}$, $\numdp=30$, so as to infer the parameter $\param$, subject to constraints $|u_t|\leq 1$.
The true parameter value is given by $\paramtr=-0.5$.
A nominal input $\unom$ is generated by sampling from $\normal{0}{0.1^2}$, which is then optimized by Alg. \ref{alg:inputdesign} using $\numq=40$ HMC samples.
The deterministic output $\ydet$ is generated using the mean values of the stochastic quantities, $w_t,v_t,x_1$.
One such application of this procedure is depicted in Fig. \ref{fig:nonlinear}.
The nominal input $\unom$, optimized input $\uopt$, and a pseudo-random binary input given by $\sign\left(\unom\right)$ are shown in Fig. \ref{fig:nonlinear}(a).
The posterior distributions $p(\param|u,\ydet(u))$ corresponding to each of these inputs are plotted in Fig. \ref{fig:nonlinear}(b).
The posterior corresponding to $\unom$ is highly concentrated about $\paramtr$.
Interestingly, the pseudo-random input leads to a posterior with greater bias and variance than even the nominal input; emphasizing that for nonlinear systems, simply scaling the `magnitude' (amplitude and/or power) of the input does not necessarily lead to better performance, in contrast to the linear case.
Fig. \ref{fig:nonlinear}(c) reports the results of maximum likelihood estimation (MLE) of $\param$, using $\uopt$ and $\sign(\unom)$ as shown in Fig. \ref{fig:nonlinear}(a).
For each of the 500 trials, data is generated by simulating the model (with a random realization of the stochastic quantities), and MLE is performed with the expectation maximization (EM) algorithm \cite{Schon2011}. As suggested by Fig. \ref{fig:nonlinear}(b), we observe a significant reduction in estimation error when using $\uopt$ compared to $\sign(\unom)$.
The results in Fig. \ref{fig:nonlinear} are representative of 
performance; 
however, exact results depend on the random realization of $\unom$.
Code to reproduce this example is available in \cite{umenberger19hmcsoftware}.

\begin{figure}
	\centering
	\subfloat[Nominal, optimized, and pseudo-random inputs.]{\includegraphics[width=0.9\columnwidth]{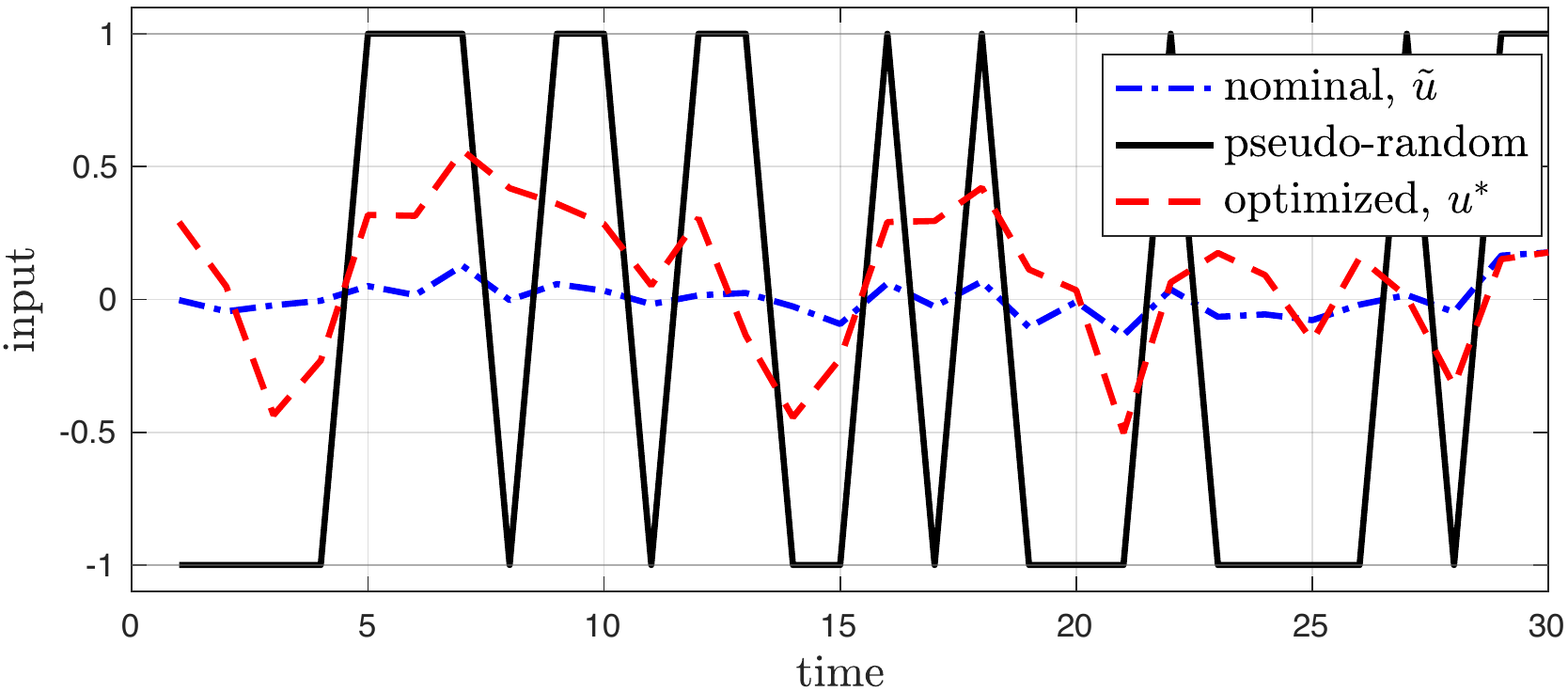}} \\
	\subfloat[$p(\param|u,\ydet(u))$ corresponding to the inputs in (a).]{\includegraphics[width=0.45\columnwidth]{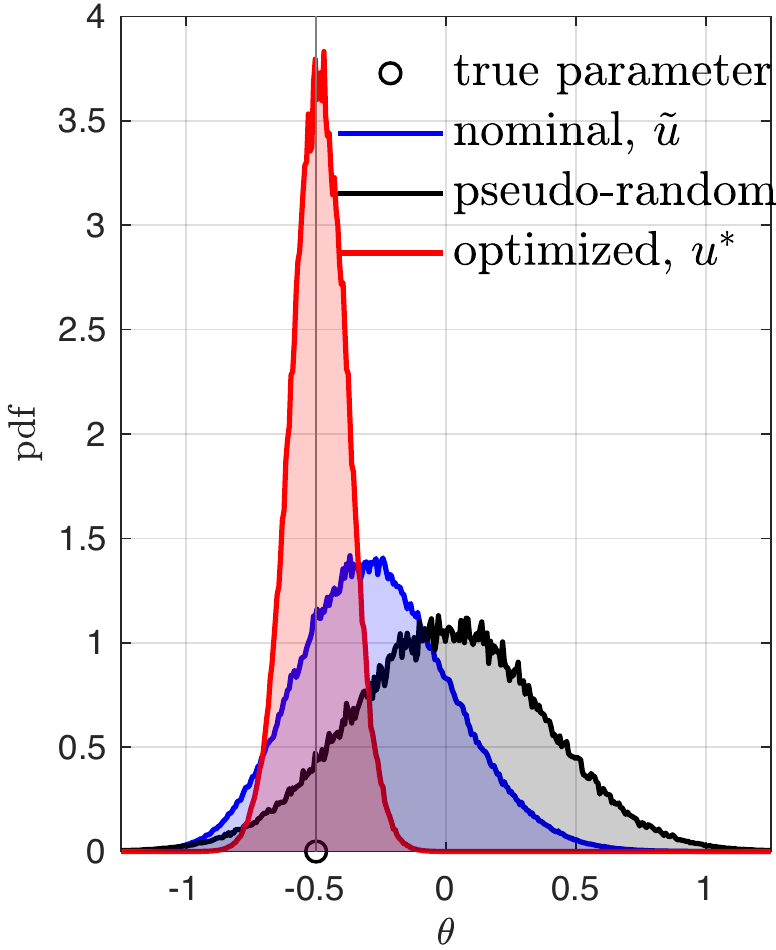}}
	\subfloat[Maximum likelihood estimation error for 500 trials.]{\includegraphics[width=0.45\columnwidth]{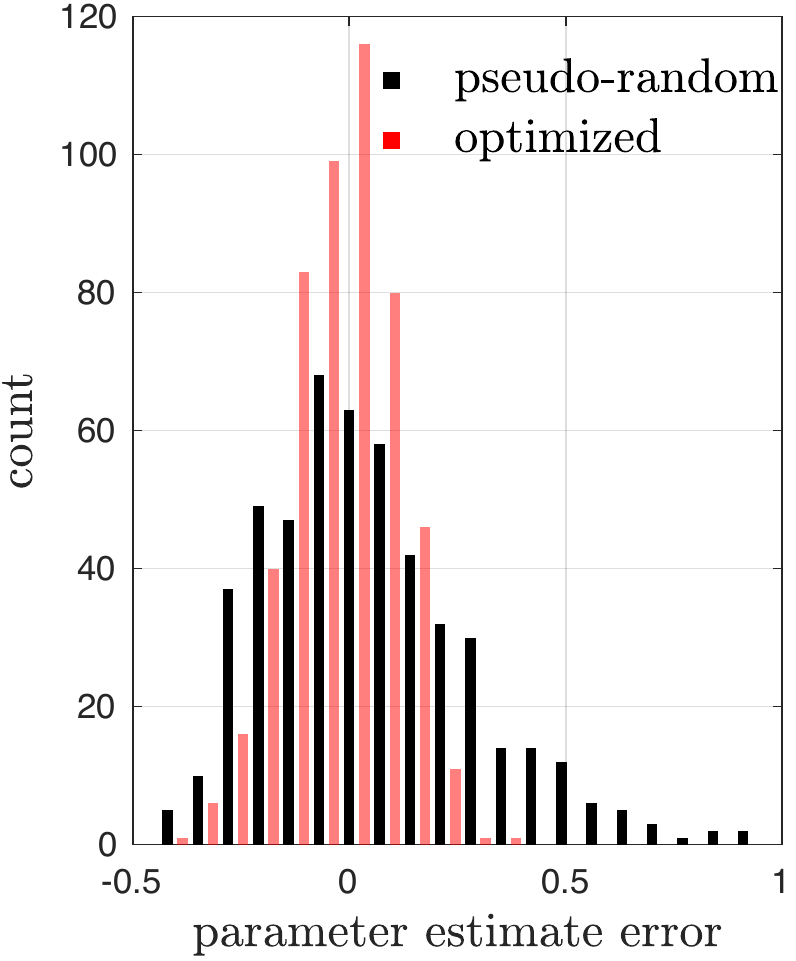}} 
	\caption{Nonlinear input design; cf. \mysec\ref{sec:ex_nl} for details.}
	\label{fig:nonlinear}
\end{figure}

\subsection{Optimal flip angles for MRI}\label{sec:ex_mri}
In this section, we consider a magnetic resonance imaging (MRI) pulse sequence design problem, identical to that studied in \cite{maidens2016parallel}.
We provide a minimal description of the design task; for full details and problem motivation, cf. \cite{maidens2016parallel}.
The dynamics of the system are given by
\begin{align*}
&x_{t+1} = \left[\begin{array}{cc}
\cos(u_t) & -\sin(u_t) \\ \sin(u_t) & \cos(u_t)
\end{array}\right]
\left[\begin{array}{c}
\mriparam_1x_{1,t} + 1-\mriparam_1 \\  \mriparam_2x_{2,t}
\end{array}\right], 
\end{align*}
\begin{align*}
&p(y_t|x_t,u_t,\mriparam) = \frac{y_t}{\sigma^2}\exp\left(-\frac{y_t^2+x_{2,t}^2}{2\sigma^2}\right)\mbf{0}\left(\frac{y_tx_{2,t}}{\sigma^2}\right),
\end{align*}
where $\mbf{0}$ denotes the modified Bessel function of the first kind (order zero).
The hidden states $x_t\in\real^2$ correspond to magnetization. 
The observed outputs $y_t$ are Rician-distributed, with $\sigma^2=0.1$. 
The inputs $u_t$ describe flip-angles in the RF excitation pulses, and as such they are constrained to 
$|u_t|\leq \pi$. 
The model parameters, $\mriparam$, are related to the $T_1$ and $T_2$ relaxation times by $\mriparam_i=\exp(-\Delta_t/T_i)$, where $\Delta_t=0.2$ is the sampling time (between RF pulses).
True parameters correspond to $T_1$ and $T_2$ relaxation times of 0.68 and 0.09, respectively.
The task is to design a 
sequence $u_{1:\numdp}$, $\numdp=29$, so as to estimate $\param=\mriparam_1$ given that $\mriparam_2$ is known. 

To apply our method, we first randomly generate a nominal input $\unom_t\sim\normal{0}{(\pi/9)^2}$ (the variance is chosen such that the nominal input is `small'). 
We then proceed with \myalg\ref{alg:inputdesign}, with $\numq=40$ HMC samples.
The deterministic output $\ydet_t(u)$ is generated by taking the expected value of $y_t$ given $x_t$ (note that $x_t$ is a deterministic function of $u$).
We repeat this process for 10 random initializations (of $\unom$) and report the results in Fig. \ref{fig:mri}.
Fig. \ref{fig:mri}(a) shows the posterior distributions $p(\param|u,\ydet(u))$ for various inputs $u$, including that obtained by dynamic programing (DP), $\udp$, as proposed in \cite{maidens2016parallel}.
We make two observations: 
first, the optimized input $\uopt$ demonstrates a significant improvement over the nominal input $\unom$.
Second, in all but one of the trials, $\uopt$ leads to a posterior that is more concentrated about $\paramtr$, compared to the input from DP $\udp$.
To investigate this further, we perform maximum likelihood estimation of~$\param$ using these inputs. 
Fig. \ref{fig:mri}(c) reports estimation error for 1000 such trials; note that the outputs for these trials were randomly generated from the Ricain distribution.
The histogram on the left of Fig. \ref{fig:mri}(c) compares  $\uopt$ as shown in Fig. \ref{fig:mri}(b) to $\udp$.
As suggested by Fig. \ref{fig:mri}(a), $\uopt$ leads to a reduction in the variance of the estimation error.
The histogram on the right of Fig. \ref{fig:mri} randomly selects one of the optimized inputs from Fig. \ref{fig:mri}(a) for each (MLE) trial; we also observe a reduction in the error variance, though the effect is less pronounced.

\begin{figure}
	\centering
	\subfloat[Posterior distributions $p(\param|u,\ydet(u))$ for 10 random initial $\unom$.]{\includegraphics[width=0.9\columnwidth]{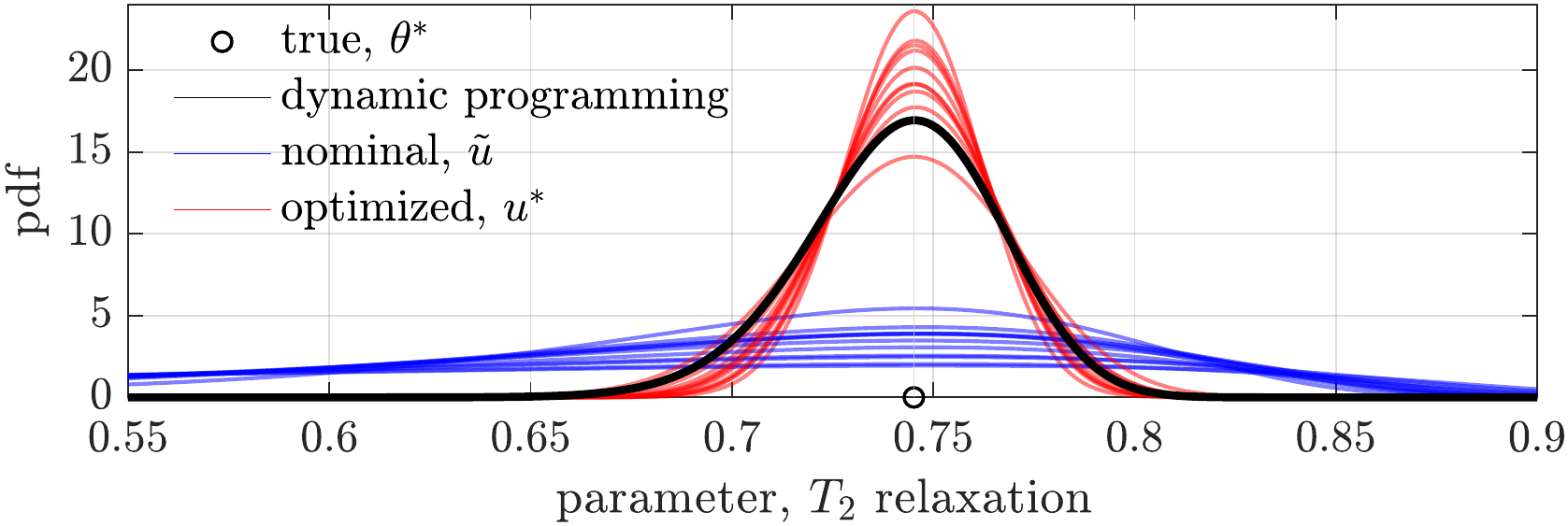}}\\
		\subfloat[Optimized input leading to most peaked posterior in (a).]{\includegraphics[width=0.9\columnwidth]{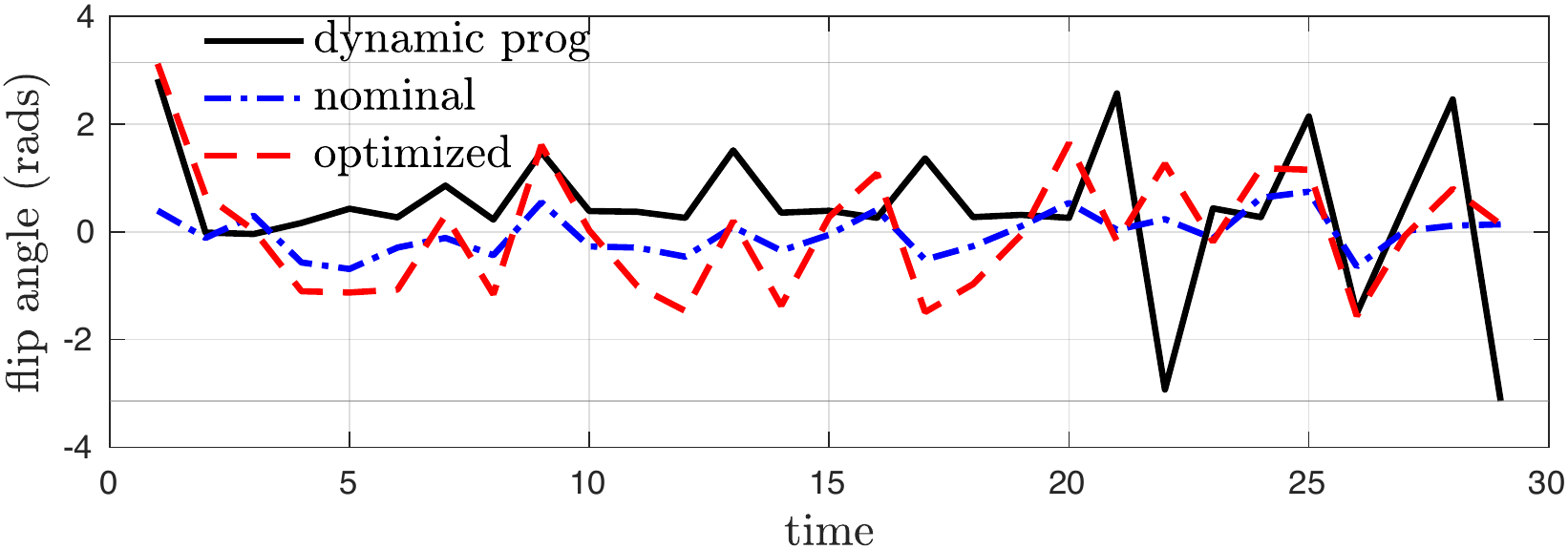}}\\
	\subfloat[Maximum likelihood error for 1000 trials. Left: optimal input in (b). Right: optimal input randomly selected from among the 10 shown in (a).]{\includegraphics[width=\columnwidth]{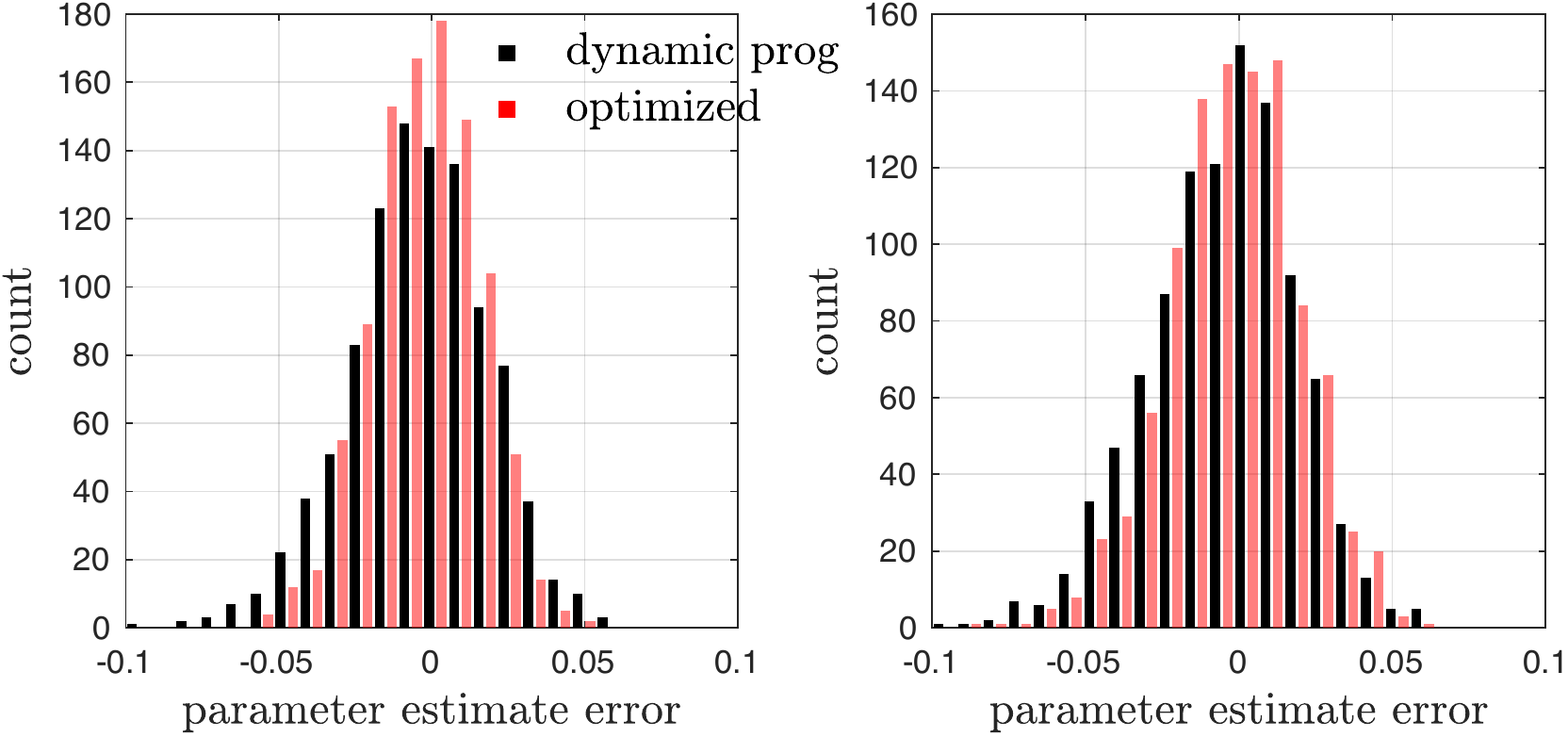}} 
	\caption{MRI pulse sequence design; cf. \mysec\ref{sec:ex_mri} for details.}
	\label{fig:mri}
\end{figure}


\bibliography{main}

\end{document}